\documentclass[10pt, conference, compsocconf]{IEEEtran}
\pdfoutput=1

\usepackage[bookmarks=false]{hyperref}
\usepackage{amssymb}
\usepackage{mathtools}
\usepackage{amsthm}
\usepackage{dblfloatfix}
\usepackage{eucal}
\usepackage{dsfont}
\usepackage[T1]{fontenc}
\usepackage[stretch=10,shrink=10]{microtype}
\usepackage[capitalise]{cleveref}

\allowdisplaybreaks[1]

\theoremstyle{plain}
\newtheorem{theorem}{Theorem}[section]
\newtheorem{lemma}[theorem]{Lemma}
\newtheorem{prop}[theorem]{Proposition}

\newtheorem*{maintheorem}{\Cref{thm:main}}
\theoremstyle{definition}
\newtheorem{definition}[theorem]{Definition}
\newtheorem{fact}[theorem]{Fact}

\newcommand {\br} [1] {\ensuremath{ \left( #1 \right) }}
\newcommand {\Br} [1] {\ensuremath{ \left[ #1 \right] }}
\newcommand {\cbr} [1] {\ensuremath{ \left\lbrace #1 \right\rbrace }}
\newcommand {\minusspace} {\: \! \!}
\newcommand {\smallspace} {\: \!}
\newcommand {\fn} [2] {\ensuremath{ #1 \minusspace \br{ #2 } }}
\newcommand {\Fn} [2] {\ensuremath{ #1 \minusspace \Br{ #2 } }}
\newcommand {\defeq} {\ensuremath{ \stackrel{\mathrm{def}}{=} }}
\newcommand {\mutinf} [2] {\fn{\mathrm{I}}{#1 : #2}}
\newcommand {\condmutinf} [3] {\mutinf{#1}{#2 \smallspace \middle\vert \smallspace #3}}
\newcommand {\prob} [1] {\Fn{\Pr}{#1}}
\newcommand {\abs} [1] {\ensuremath{ \left| #1 \right| }}
\newcommand {\norm} [1] {\ensuremath{ \left\| #1 \right\| }}
\newcommand {\normsub} [2] {\ensuremath{ \norm{#1}_{#2} }}
\newcommand {\onenorm} [1] {\normsub{#1}{1}}
\newcommand {\onenormbigg} [1] {\ensuremath{ \bigg\| #1 \bigg\|_{1} }}
\newcommand {\relent} [2] {\fn{\mathrm{S}}{#1 \middle\| #2}}
\newcommand {\rminent} [2] {\fn{\mathrm{S}_{\infty}}{#1 \middle\| #2}}
\DeclareMathOperator*{\bigE}{\mathbb{E}}
\newcommand {\expec} [2] {\Fn{\bigE_{\substack{#1}}}{#2}}
\newcommand {\expecbig} [2] {\ensuremath{ \bigE_{\substack{#1}} \minusspace \big[ #2 \big] }}
\newcommand {\email} [1] {\href{mailto:#1}{#1}}

\newcommand {\bigomega} [1] {\fn{\Omega}{#1}}
\newcommand {\ie} {i.e.,}
\newcommand {\finset} [1] {\ensuremath{\CMcal{#1}}}
\newcommand {\bra} [1] {\ensuremath{ \left\langle #1 \right| }}
\newcommand {\ket} [1] {\ensuremath{ \left| #1 \right\rangle }}
\newcommand {\ketbratwo} [2] {\ensuremath{ \left| #1 \middle\rangle \middle\langle #2 \right| }}
\newcommand {\ketbra} [1] {\ketbratwo{#1}{#1}}
\newcommand {\cspace} [1] {\ensuremath{\mathnormal{#1}}}
\newcommand {\register} [1] {\ensuremath{ \mathnormal{#1} }}
\newcommand {\unitary} [1] {\ensuremath{ \mathbf{#1} }}
\newcommand {\ceil} [1] {\ensuremath{ \left\lceil #1 \right\rceil }}
\newcommand {\floor} [1] {\ensuremath{ \left\lfloor #1 \right\rfloor }}
\newcommand {\Tsr} {\ensuremath{ T^{\br{r}} }}
\newcommand {\Tr} {\ensuremath{ \mathrm{Tr} }}
\newcommand {\partrace} [2] {\fn{\Tr_{#1}}{#2}}
\newcommand {\id} {\ensuremath{\mathds{1}}}
\newcommand {\comp} [1] {\bar{#1}}
\newcommand {\xtildesp} {\tilde{\cspace{X}}}
\newcommand {\ytildesp} {\tilde{\cspace{Y}}}
\newcommand {\xtildereg} {\tilde{\register{X}}}
\newcommand {\ytildereg} {\tilde{\register{Y}}}
\newcommand {\finalst} {\ensuremath{\varphi}}
\newcommand {\startingst} {\ensuremath{\theta}}
\newcommand {\adjoint} [1] {\ensuremath{#1^{*}}}
\newcommand {\Bob} {\ensuremath{\tilde{\register{B}}}}
\newcommand {\Alice} {\ensuremath{\tilde{\register{A}}}}
\newcommand {\Prot}{\mathcal{P}}
\newcommand {\suppress}[1]{}

\newcommand {\set} [1] {\ensuremath{ \left\lbrace #1 \right\rbrace }}
\def\ve{{\varepsilon}}
\def\F{\mathrm{F}}
\newcommand {\reg} [1] {\ensuremath{ \mathnormal{#1} }}

\def\E{\mathcal{E}}

\newcommand {\mytitle} {A parallel repetition theorem for entangled two-player one-round games
	under product distributions}
\newcommand {\Rahul}   {Rahul Jain}
\newcommand {\Attila}  {Attila Pereszl\'{e}nyi}
\newcommand {\Penghui} {Penghui Yao}

\newcommand {\CQT} {Centre for Quantum Technologies}
\newcommand {\CQTCS} {\CQT{} and\\ Department of Computer Science}
\newcommand {\NUS} {National University of Singapore}

\hypersetup{
	pdfstartview={FitH},
	pdfdisplaydoctitle={true},
	breaklinks={true},
	pdftitle={\mytitle},
	pdfauthor={\Rahul, \Attila, \Penghui}
}

\newcommand {\BigCalculation} {
\begin{figure*}[!b]
\normalsize
\vspace*{13pt}
\hrulefill
\begin{align}
	& \onenorm{\expec{(x,y)\leftarrow \mu}{\ketbra{xy}\otimes\ketbra{\varphi_{x,y}}}-
		\expec{(x,y)\leftarrow \mu_X\otimes\mu_Y}{\ketbra{xy}\otimes\br{\unitary{U}_x \otimes \unitary{V}_y}
		\ketbra{\varphi}
		\br{\adjoint{\unitary{U}_x} \otimes  \adjoint{\unitary{V}_y}}}} \nonumber \\
	& \quad \leq \onenorm{\expec{(x,y)\leftarrow \mu}{\ketbra{xy}\otimes\ketbra{\varphi_{x,y}}}
		- \expec{(x,y)\leftarrow \mu_X \otimes \mu_Y}
		{\ketbra{xy}\otimes\br{\unitary{U}_x \otimes \id_{\cspace{B}}}
		\ketbra{\varphi_y}
		\br{\adjoint{\unitary{U}_x} \otimes \id_{\cspace{B}}}}} \nonumber \\
	& \quad \quad {} + \onenorm{\expec{(x,y)\leftarrow \mu_X \otimes \mu_Y}
		{\ketbra{xy}\otimes\br{\unitary{U}_x \otimes \id_{\cspace{B}}}
		\ketbra{\varphi_y}
		\br{\adjoint{\unitary{U}_x} \otimes \id_{\cspace{B}}}
		- \ketbra{xy}\otimes\br{\unitary{U}_x \otimes \unitary{V}_y}
		\ketbra{\varphi}
		\br{\adjoint{\unitary{U}_x} \otimes \adjoint{\unitary{V}_y}}}}
		\label{eqn:triangle inequality} \\
	& \quad \leq \onenorm{\expec{x\leftarrow\mu_X}{\ketbra{x}\otimes\ketbra{\varphi_{x}} - \ketbra{x}\otimes
		\br{\unitary{U}_x \otimes \id_{\cspace{B}}}
		\ketbra{\varphi}
		\br{\adjoint{\unitary{U}_x} \otimes \id_{\cspace{B}}}}} \nonumber \\
	& \quad \quad {} +
		\onenorm{\expec{xy\leftarrow\mu_X\otimes\mu_Y} { \ketbra{xy} \otimes \ketbra{\varphi_y} - \ketbra{xy} \otimes
		\br{\id_{\cspace{A}} \otimes \unitary{V}_y}
		\ketbra{\varphi}
		\br{\id_{\cspace{A}} \otimes\adjoint{\unitary{V}_y}}}}
		\label{eqn:Ux} \\
	& \quad = \expec{x\leftarrow\mu_X}{
		\onenorm{\ketbra{\varphi_{x}} -
		\br{\unitary{U}_x \otimes \id_{\cspace{B}}}
		\ketbra{\varphi}
		\br{\adjoint{\unitary{U}_x} \otimes \id_{\cspace{B}}}}} \nonumber \\
	& \quad \quad {} + \expec{y\leftarrow\mu_Y}{
		\onenorm{\ketbra{\varphi_y} -
		\br{\id_{\cspace{A}} \otimes \unitary{V}_y}
		\ketbra{\varphi}
		\br{\id_{\cspace{A}} \otimes\adjoint{\unitary{V}_y}}}} \nonumber \\
	& \quad \leq 8 \sqrt{\ve} \label{eqn:intermediate bound}
\end{align}
\end{figure*}
}

\begin{document}

\title{\mytitle}

\author{\IEEEauthorblockN{\Rahul}
\IEEEauthorblockA{\CQTCS\\
	\NUS\\
	Singapore\\
	\email{rahul@comp.nus.edu.sg}}
\and
\IEEEauthorblockN{\Attila}
\IEEEauthorblockA{\CQT\\
	\NUS\\
	Singapore\\
	\email{attila.pereszlenyi@gmail.com}}
\and
\IEEEauthorblockN{\Penghui}
\IEEEauthorblockA{\CQT\\
	\NUS\\
	Singapore\\
	\email{phyao1985@gmail.com}}
}

\maketitle

\begin{abstract}
We show a {\em parallel repetition} theorem for the {\em entangled} value
$\omega^*(G)$ of any {\em two-player one-round game} $G$ where the questions
$(x,y) \in \finset{X}\times\finset{Y}$ to Alice and Bob are drawn from a
product distribution on $\finset{X}\times\finset{Y}$.
We show that for the $k$-fold product $G^k$ of the game $G$
(which represents the game $G$ played in parallel $k$ times independently)
\[ \fn{\omega^*}{G^k} = \br{1-(1-\omega^*(G))^3}^{\bigomega{\frac{k}
{\log(|\finset{A}| \cdot |\finset{B}|)}}} \]
 where $\finset{A}$ and $\finset{B}$ represent the sets from which the answers of Alice and Bob are drawn.

The arguments we use are information theoretic and are broadly on similar lines as that of  Raz~\cite{Raz:1995:PRT:225058.225181} and Holenstein~\cite{Holenstein2007} for classical games.
The additional quantum ingredients we need, to deal with entangled games, are inspired by the work
of Jain, Radhakrishnan, and Sen~\cite{Jain:2008}, where quantum information theoretic arguments were
used to achieve message compression in quantum communication protocols.
\end{abstract}


\section{Introduction}

A {\em two-player one-round game} $G$ is specified by finite sets \finset{X}, \finset{Y}, \finset{A}, and \finset{B}, a distribution $\mu$ over $\finset{X}\times\finset{Y}$, and a predicate $V:\finset{X}\times \finset{Y}\times\finset{A}\times\finset{B}\rightarrow\set{0,1}$. It is played as follows.

\begin{itemize}

\item The referee selects questions $(x,y)\in\finset{X}\times\finset{Y}$ according to distribution $\mu$.

\item The referee sends $x$ to Alice and $y$ to Bob. Alice and Bob are spatially separated, so they do not see each other's input.

\item Alice chooses answer $a\in\finset{A}$ and sends it back to the referee. Bob chooses answer $b\in\finset{B}$ and sends it back to the referee.

\item The referee accepts if $V(x,y,a,b)=1$ and otherwise rejects. Alice and Bob win the game if the referee accepts.

\end{itemize}
The {\it value} of the game $G$, denoted by $\omega(G)$, is defined to be the maximum winning probability (averaged over the distribution $\mu$) achieved by Alice and Bob.

These games have played an important and pivotal role in the study of the rich theory of {\em inapproximability}, leading to the development of {\em Probabilistically Checkable Proofs} and the famous {\em Unique Games Conjecture}.  One of the most fundamental problems regarding this model is the so called
{\em parallel repetition} question, which concerns  the behavior of multiple copies of the game played in parallel. For the  game $G=(\mu,\finset{X}, \finset{Y}, \finset{A}, \finset{B},V)$, its  $k$-fold product is given by $G^k=(\mu^k,\finset{X}^k, \finset{Y}^k, \finset{A}^k, \finset{B}^k,V^k)$, where $V^k(x,y,a,b)=1$ if and only if $V(x_i,y_i,a_i,b_i)=1$ for all $i\in[k]$. Namely, Alice and Bob play $k$ copies of game $G$ simultaneously, and they win  iff they win all the copies. It is easily seen that $\omega(G^k)\geq\omega(G)^k$ for any game $G$.
The equality of the two quantities, for all games, was conjectured by Ben-Or, Goldwasser,
Kilian and Wigderson~\cite{Ben-Or:1988:MIP:62212.62223}. The conjecture was shown to be false by Fortnow~\cite{Fortnow:1989}.

However one could still expect that $\omega(G^k)$ goes down exponentially in $k$ (asymptotically).
This is referred to as the parallel repetition (also known as the {\em direct product}) conjecture.
This was shown to be indeed true in a seminal paper by Raz~\cite{Raz:1995:PRT:225058.225181}.
Raz showed that
\[ \fn{\omega}{G^k} = (1-(1-\omega(G))^c)^{\bigomega{\frac{k}{\log(|\finset{A}||\finset{B}|)}}} \]
where $c$ is a universal constant. This result, along with the the {\em PCP theorem} had deep consequences for the theory of inapproximability~\cite{Arora:1998:PCP:273865.273901,Arora:1998:PVH:278298.278306,Dinur:2007:PTG:1236457.1236459}. A series of works later exhibited improved results for general and specific games~\cite{Holenstein2007,Rao2008,Raz2008,Barak:2009:SPR:1616497.1616526,Raz:2012}.

In the quantum setting, it is natural to consider the so called {\em entangled games}
where Alice and Bob are, in addition, allowed to share a quantum state before the games starts. The questions and answers in the game remain classical.
On receiving questions, Alice and Bob can generate their answers by making quantum measurements on their shared entangled state. The value of the entangled version of the game $G$ is denoted by $\omega^*(G)$. The study of  entangled games is deeply related to the foundation of quantum mechanics and that of quantum entanglement.
These games have been used to give a novel interpretation to {\it Bell inequalities},
one of the most famous and useful methods for differentiating classical and quantum mechanics
(e.g., by Clauser, Horne, Shimony and Holt~\cite{CHSH:69}).
Recently these games have also been studied from cryptographic motivations
such as in Refs.~\cite{Esther:2010,TomamichelFKW13,MasanesPA11}.
Analogously to the classical case, the study of the parallel repetition question in this setting may potentially have applications in quantum complexity theory.

The parallel repetition conjecture has been shown to be true for several sub-classes of entangled games, starting with the so called {\em XOR games} by Cleve, Slofstra, Unger and Upadhyay~\cite{Cleve:2008:PPR:1391349.1391350}, later generalized to {\em unique games} by Kempe, Regev and Toner~\cite{Kempe2010} and very recently further generalized to {\em projection games} by Dinur, Steurer and Vidick~\cite{DinurSteurerVidick:2013} (following an analytical framework introduced by Dinur and Steurer in~\cite{DinurSteurer:2013} to deal with classical projection games).
For general games, Kempe and Vidick~\cite{Kempe:2011:PRE:1993636.1993684} (following a framework by Feige and Killian~\cite{Feige:2000:TPE:586850.587040} for classical games)
showed a parallel repetition theorem albeit with only a polynomial decay in $k$, in the value $\omega^*(G^k)$.
In a recent work, Chailloux and Scarpa~\cite{Chailloux:2013}
showed an exponential decay in $\omega^*(G^k)$ using information theoretic arguments.
\begin{theorem}[\cite{Chailloux:2013}]
	\label{thm:chailloux}
For any game $G=(\mu,\finset{X},\finset{Y},\finset{A},\finset{B},V)$, where $\mu$ is the uniform distribution on $\finset{X} \times \finset{Y}$, it holds that
\[
	\fn{\omega^*}{G^k} = \br{1 - (1-\omega^*(G))^2}^{\bigomega{\frac{k}
	{\log(|\finset{A}||\finset{B}||\finset{X}||\finset{Y}|)}}}.
\]
	As a corollary, for a general distribution $\mu$,
\[
	\fn{\omega^*}{G^k} = \br{1 - (1-\omega^*(G))^2}^{\bigomega{\frac{k}
	{Q^4 \log (Q) \log(|\finset{A}||\finset{B}|)}}}
\]
	where
\[
	Q = \max \cbr{ \ceil{ \frac{1}{\min_{x,y : \mu(x,y) \neq 0}
	\cbr{\sqrt{\mu(x,y)}}}}, |\finset{X}|\cdot|\finset{Y}|} .
\]
\end{theorem}
Note that here $\omega^*(G^k)$ depends on $|\finset{X}| \cdot |\finset{Y}|$ as well, in addition to $|\finset{A}| \cdot |\finset{B}|$ (as in Raz's result). Also the value of $Q$ can be arbitrarily large, depending on the distribution $\mu$.
\subsection*{Our result}
In this paper we consider the case when the distribution $\mu$
is product across $\finset{X} \times \finset{Y}$.
That is, there are distributions $\mu_X, \mu_Y$ on $\finset{X}, \finset{Y}$ respectively such that $\forall (x,y) \in \finset{X} \times \finset{Y}: \mu(x,y) = \mu_X(x) \cdot \mu_Y(y)$.
We show the following.
\begin{theorem}[Main Result]
	\label{thm:main}
For any game \[G=(\mu,\finset{X},\finset{Y},\finset{A},\finset{B},V)\]
where $\mu$ is a product distribution on $\finset{X} \times \finset{Y}$, it holds that
\[ \fn{\omega^*}{G^k} = \br{1-(1-\omega^*(G))^3}^{\bigomega{\frac{k}
{\log(|\finset{A}||\finset{B}|)}}} .\]
\end{theorem}
Note that the uniform distribution on $\finset{X} \times \finset{Y}$
is a product distribution and our result has no dependence on the size
of $\finset{X} \times \finset{Y}$.
Hence, our result implies and strengthens on the result of
Chailloux and Scarpa~\cite{Chailloux:2013} (up to the exponent of $1-\omega^*(G)$).

\subsection*{Our techniques}

The arguments we use are information theoretic and are broadly on similar lines as that of  Raz~\cite{Raz:1995:PRT:225058.225181} and Holenstein~\cite{Holenstein2007} for classical games.
The additional quantum ingredients we need, to deal with entangled games, are inspired by the work
of Jain, Radhakrishnan, and Sen~\cite{Jain:2008}, where quantum information theoretic arguments were
used to achieve message compression in quantum communication protocols.

Given the $k$-fold game $G^k$,
let us condition on success on a set $\finset{C}\subseteq[k]$ of coordinates.
If the overall success in coordinates in \finset{C} is already as small as we want,
then we are done.
Otherwise, we exhibit another coordinate $j\notin \finset{C}$
such that the success in the $j$-th coordinate,
even when conditioning on success in the coordinates inside \finset{C},
is bounded away from $1$.
Here we assume that $\omega^*(G)$ is bounded away from $1$.
This way the overall success keeps going down and becomes exponentially small in $k$, after we have identified $\Omega(k)$ such coordinates.
To argue that  the probability with which Alice and Bob win the $j$-th coordinate,
conditioned on success in \finset{C}, is bounded away from $1$,
we show that close to this success probability can be achieved for a single instance of the game $G$.
That is, given inputs $(x',y')$, drawn from $\mu$, for a single instance of $G$,
Alice and Bob can embed $(x',y')$ to the $j$-th coordinate of $G^k$, conditioned on success in \finset{C},
and generate the rest of the state with good approximation.
So, if the probability of success in the $j$-th coordinate, conditioned on success in \finset{C},
is very close to $1$, there is a strategy for $G$ with probability of success strictly larger
than $\omega^*(G)$, reaching a contradicting to the definition of $\omega^*(G)$.

Suppose the global state, conditioned on success in \finset{C}, is of the form
\[ \sigma^{XYAB} = \sum_{x \in \finset{X}^k, y \in \finset{Y}^k}
\tilde{\mu}(x,y) \ketbra{xy}^{\register{X}\register{Y}}
\otimes \ketbra{\phi_{xy}}^{\register{A}\register{B}} \]
where $\tilde{\mu}$ is a distribution, potentially different from
$\mu$ because of the conditioning on success.
(Here we further fix the questions and answers in \finset{C} to specific values  and do not specify them in $\sigma^{XYAB}$.)
In protocol $\Prot$ for the single instance of $G$,
we let Alice and Bob start with the shared pure state
\[ \ket{\finalst} = \sum_{x \in \finset{X}^k, y \in \finset{Y}^k} \sqrt{\tilde{\mu}(x,y)}
\ket{xxyy}^{\xtildereg\register{X}\ytildereg\register{Y}} \otimes
\ket{\phi_{xy}}^{\register{A}\register{B}} . \]
Note that \ket{\finalst} is a purification of $\sigma^{XYAB}$, where registers
$\xtildereg$ and $\ytildereg$ are identical to \register{X} and \register{Y}. We introduce these copies of the registers \register{X} and \register{Y} so that the marginal state
in these registers remains a {\em classical state} and these registers can be viewed as classical registers, which is important in our arguments.

\suppress{The reason we want the global state, including the questions, answers, and
the shared state, to be pure is because we can then use properties of
purifications such as the unitary equivalence of purifications
and Uhlmann's theorem.
This trick has been used before in other papers \cite{Jain:2003:LBB:946243.946331,Jain:2008}.
It has also been used recently by Chailloux and Scarpa~\cite{Chailloux:2013}.}

Using the chain rule for mutual information, we are able to argue that
both $\mutinf{X_j}{Y\tilde{Y}B}$ and $\mutinf{Y_j}{X\tilde{X}A}$
are very small (close to $0$), in $\ket{\finalst}$.
This, obviously, is only possible when the distribution
$\mu$ is product.
In addition, the distribution of the questions in the $j$-th coordinate,
in $\ket{\finalst}$, remains close to $\mu$, in the $\ell_1$-{\em distance}.
In protocol $\Prot$, when Alice and Bob get questions $x'$ and $y'$,
suppose they measure registers $\register{X}_j$ and $\register{Y}_j$,
in \ket{\finalst}, and get $x'_j$ and $y'_j$.
Let \ket{\finalst_{x'_jy'_j}} be the resulting state.
If by luck it so happens that $\br{x',y'} = \br{x'_j,y'_j}$,
then they can measure the answer registers $\register{A}_j$ and $\register{B}_j$,
in \ket{\finalst_{x'_jy'_j}}, respectively, and send the answers to the referee.
However, the probability that $(x',y') = \br{x'_j,y'_j}$ can be very small and they
want to get this desired outcome with probability very close to $1$.
We describe next how this can be achieved.

Let \ket{\finalst_{x'_j}} be the resulting state obtained after we measure register
$\register{X}_j$ (in $|\finalst\rangle$) and obtain outcome $x'_j$.
The fact that $\mutinf{X_j}{Y\tilde{Y}B}$ is close to $0$
implies that Bob's side of \ket{\finalst_{x'_j}} is mostly independent of $x'_j$.
By the unitary equivalence of purifications and Uhlmann's theorem,
there is a unitary transformation $\unitary{U}_{x_j'}$
that Alice can apply to take the state  $\ket{\finalst}$ quite close
to the state $\ket{\finalst_{x_j'}}$.
Similarly, let us define \ket{\finalst_{y'_j}} and again
$\mutinf{Y_j}{X\tilde{X}A}$ being close to $0$ implies that Alice's side of
\ket{\finalst_{y'_j}} is mostly independent of $y'_j$.
Again, by Uhlmann's theorem, there is a unitary transformation $\unitary{U}_{y_j'}$
that Bob can apply to take the state $\ket{\finalst}$ quite close
to the state $\ket{\finalst_{y_j'}}$.
Interestingly, as was argued in~\cite{Jain:2008}, when Alice and Bob simultaneously
apply $\unitary{U}_{x_j'}$ and $\unitary{U}_{y_j'}$,
they take $\ket{\finalst}$ quite close to the state $\ket{\finalst_{x_j' y_j'}}$!
This again requires the distribution of questions to be
independent across Alice and Bob.

\subsection*{Organization of the paper}

In \cref{sec:Preliminaries}, we present some background on information
theory, as well as some useful lemmas that we will need for our proof.
In \cref{sec:main result}, we prove our main result, \cref{thm:main}.

\section{Preliminaries}
\label{sec:Preliminaries}

In this section we present some notations, definitions, facts,
and lemmas that we will use later in our proof.

\subsection*{Information theory}

For integer $n \geq 1$, let $[n]$ represent the set $\{1,2, \ldots, n\}$.
Let $\finset{X}$ and $\finset{Y}$ be finite sets and $k$ be a natural number.
Let $\finset{X}^k$ be the set $\finset{X}\times\cdots\times\finset{X}$, the cross product of
$\finset{X}$, $k$ times.
Let $\mu$ be a probability distribution on $\finset{X}$.
Let $\mu(x)$ represent the probability of $x\in\finset{X}$ according to $\mu$.
Let $X$ be a random variable distributed according to $\mu$.
We use the same symbol to represent a
random variable and its distribution whenever it is clear from the
context.
The expectation value of function $f$ on $\finset{X}$ is defined as
$\expec{x \leftarrow X}{f(x)} \defeq \sum_{x \in \finset{X}} \prob{X=x}
\cdot f(x)$, where $x\leftarrow X$ means that $x$ is drawn from the distribution of $X$.
A quantum state (or just a state) $\rho$ is a positive semi-definite matrix with trace equal to $1$. It is pure if and only if the rank is $1$. Let $\ket{\psi}$ be a unit vector.
With slight abuse of notation, we use $\psi$ to represent the state
and also the density matrix  $\ketbra{\psi}$, associated with $\ket{\psi}$.
A classical distribution $\mu$ can be viewed as a quantum state with diagonal entries $\mu(x)$ and non-diagonal entries $0$.
For two quantum states $\rho$ and $\sigma$, $\rho\otimes\sigma$ represents the tensor product (Kronecker product) of $\rho$ and $\sigma$.
A quantum super-operator $\E(\cdot)$ is a completely positive and trace preserving (CPTP) linear map from states to states.
Readers can refer to~\cite{CoverT91,NielsenC00,Watrouslecturenote} for more details.

\begin{definition}\label{def:tracedistance}
For quantum states $\rho$ and $\sigma$, the  $\ell_1$-distance between them is given by $\onenorm{\rho-\sigma}$, where $\onenorm{X}\defeq\Tr\sqrt{X^{\dag}X}$ is the sum of the singular values of $X$. We say that $\rho$ is $\varepsilon$-close to $\sigma$ if $\|\rho-\sigma\|_1\leq\varepsilon$.
\end{definition}
\begin{definition}\label{def:fidelity}
For quantum states $\rho$ and $\sigma$, the {\em fidelity} between them is given by $\F(\rho,\sigma)\defeq\onenorm{\sqrt{\rho}\sqrt{\sigma}}.$
\end{definition}

The following proposition states that the distance between
two states can't be increased by quantum operations.
\begin{prop}[\cite{NielsenC00}, pages 406 and 414]
	\label{prop:monotonequantumoperation}
For states $\rho$, $\sigma$, and quantum operation $\E(\cdot)$, it holds that
\begin{align*}
	\onenorm{\E(\rho) - \E(\sigma)} &\leq \onenorm{\rho - \sigma}
	\intertext{and}
	\F(\E(\rho),\E(\sigma)) &\geq \F(\rho,\sigma) .
\end{align*}
\end{prop}

The following proposition relates the $\ell_1$-distance and the fidelity between two states.

\begin{prop}[\cite{NielsenC00}, page 416]
	\label{prop:tracefidelityequi}
	For quantum states $\rho$ and $\sigma$, it holds that
	\[2(1-\F(\rho,\sigma))\leq\onenorm{\rho-\sigma}\leq2\sqrt{1-\F(\rho,\sigma)^2}.\]
For two pure states $\ket{\phi}$ and $\ket{\psi}$, we have
\begin{align*}
	\onenorm{\ketbra{\phi} - \ketbra{\psi}}
	&= \sqrt{1 - \fn{\F}{\ketbra{\phi}, \ketbra{\psi}}^2} \\
	&= \sqrt{1 - |\langle\phi|\psi\rangle|^2}.
\end{align*}
\end{prop}

Let $\rho^{AB}$ be a bipartite quantum state in registers $AB$. We use the same symbol to represent a quantum register and the Hilbert space associated with it.
We define
\[ \rho^{\reg{B}} \defeq \partrace{\reg{A}}{\rho^{AB}}
\defeq \sum_i (\bra{i} \otimes \id_{\cspace{B}})
\rho^{AB} (\ket{i} \otimes \id_{\cspace{B}}) \]
where $\set{\ket{i}}_i$ is a basis for the Hilbert space $\cspace{A}$
and $\id_{\cspace{B}}$ is the identity matrix in space $\cspace{B}$.
The state $\rho^B$ is referred to as the marginal state of $\rho^{AB}$ in register $B$.
\begin{definition}\label{def:purification}
	We say that a pure state $\ket{\psi} \in \cspace{A}\otimes\cspace{B}$
	is a purification of some state $\rho$
	if $\Tr_{\cspace{A}}(\ketbra{\psi})=\rho$.
\end{definition}

\begin{theorem}[Uhlmann's theorem]
	\label{thm:Uhlmann}
	Given quantum states $\rho$, $\sigma$, and a purification $\ket{\psi}$ of $\rho$,
	it holds that $\F(\rho,\sigma)=\max_{\ket{\phi}} | \langle \phi | \psi \rangle | $,
	where the maximum is taken over all purifications of $\sigma$.
\end{theorem}

The {\em entropy} of a quantum state $\rho$ (in register $\reg{X}$) is defined as $\mathrm{S}(\rho) \defeq - \Tr\rho\log\rho.$
We also let $\mathrm{S}\br{\reg{X}}_{\rho}$ represent $\mathrm{S}(\rho)$.
The {\em relative entropy} between quantum states $\rho$ and $\sigma$ is defined as $\relent{\rho}{\sigma} \defeq \Tr\rho\log\rho-\Tr\rho\log\sigma.$
The {\em relative min-entropy} between them is defined as
$\rminent{\rho}{\sigma} \defeq \min\set{\lambda : \rho\leq2^\lambda\sigma}$.
Since the logarithm is operator-monotone, $\relent{\rho}{\sigma} \leq \rminent{\rho}{\sigma}$.
Let $\rho^{\reg{X}\reg{Y}}$ be a quantum state in space $\cspace{X}\otimes\cspace{Y}$. The {\em mutual information} between registers $\reg{X}$ and $\reg{Y}$ is defined to be
\begin{align*}
	\mutinf{\reg{X}}{\reg{Y}}_{\rho} &\defeq
	\mathrm{S}\br{\reg{X}}_{\rho}+\mathrm{S}\br{\reg{Y}}_{\rho}-\mathrm{S}\br{\reg{X}\reg{Y}}_{\rho}.
\end{align*}
It is easy to see that $\mutinf{\reg{X}}{\reg{Y}}_{\rho}
= \relent{\rho^{\reg{XY}}}{\rho^{\reg{X}} \otimes \rho^{\reg{Y}}}$.
If $\reg{X}$ is a classical register, namely
$\rho = \sum_x \mu(x) \ketbra{x} \otimes \rho_x$,
where $\mu$ is a probability distribution over $\reg{X}$, then
\begin{align*}
	\mutinf{\reg{X}}{\reg{Y}}_{\rho}
	&= \mathrm{S}\br{\reg{Y}}_{\rho}-\mathrm{S}\br{\reg{Y}|\reg{X}}_{\rho} \\
	&= \mathrm{S}\br{\sum_x\mu(x)\rho_x}-\sum_x\mu(x)\mathrm{S}\br{\rho_x}
\end{align*}
where the {\em conditional entropy} is defined as
\[\mathrm{S}(\reg{Y}|\reg{X})_{\rho}\defeq\expec{x\leftarrow\mu}{\mathrm{S}(\rho_x)}.\]
Let $\rho^{\reg{X}\reg{Y}\reg{Z}}$ be a quantum state with $\reg{Y}$ being a classical register.
The mutual information between $\reg{X}$ and $\reg{Z}$, conditioned on
$Y$, is defined as
\begin{align*}
	\condmutinf{\reg{X}}{\reg{Z}}{\reg{Y}}_{\rho} &\defeq
	\expec{y \leftarrow Y}{\condmutinf{\reg{X}}{\reg{Z}}{Y=y}_{\rho}} \\
	&= \mathrm{S}\br{\reg{X}|\reg{Y}}_{\rho} +
	\mathrm{S}\br{\reg{Z}|\reg{Y}}_{\rho} -
	\mathrm{S}\br{\reg{X}\reg{Z}|\reg{Y}}_{\rho} .
\end{align*}
The following {\em chain rule} for mutual information follows
easily from the definitions, when $\reg{Y}$ is a classical register.
\[ \mutinf{\reg{X}}{\reg{Y}\reg{Z}}_{\rho} = \mutinf{\reg{X}}{\reg{Y}}_{\rho} + \condmutinf{\reg{X}}{\reg{Z}}{\reg{Y}}_{\rho} .\]
We will need the following basic facts.
\begin{fact}
	\label{fact:relative entropy joint convexity}
	The relative entropy is
	jointly convex in its arguments.
	That is, for quantum states $\rho$,
	$\rho^1$, $\sigma$, and $\sigma^1$, and $p\in[0,1]$,
	\begin{align*}
		&\relent{p \rho + (1-p) \rho^1}{p \sigma + (1-p) \sigma^1} \\
		& \qquad \leq p \cdot \relent{\rho}{\sigma}
		+ (1-p) \cdot \relent{\rho^1}{\sigma^1} .
	\end{align*}
\end{fact}
We have the following chain rule for the relative-entropy.
\begin{fact}
	\label{fact:relative entropy splitting}
	Let
	\begin{align*}
		\rho &= \sum_x\mu(x) \ketbra{x} \otimes\rho_x
		\intertext{and}
		\rho^1 &= \sum_x\mu^1(x) \ketbra{x} \otimes\rho^1_x .
	\end{align*}
	It holds that
	\[ \relent{\rho^1}{\rho} = \relent{\mu^1}{\mu}
	+ \expec{x\leftarrow \mu^1} {\relent{\rho^1_x}{\rho_x}}.\]
\end{fact}

\begin{fact} \label{fact:mutinf is min}
    For quantum states $\rho^{\reg{X}\reg{Y}}$, $\sigma^{\reg{X}}$,
    and $\tau^{\reg{Y}}$,
    it holds that
    \[ \relent{\rho^{\reg{X}\reg{Y}}}{\sigma^{\reg{X}} \otimes \tau^{\reg{Y}}}
    \geq \relent{\rho^{\reg{X}\reg{Y}}}{\rho^{\reg{X}}\otimes \rho^{\reg{Y}}}=\mutinf{X}{Y}_{\rho}. \]
\end{fact}

\begin{fact}[\cite{Watrouslecturenote,Jain:2003:LBB:946243.946331}]
    \label{fact:one norm and relent}
    For quantum states $\rho$ and $\sigma$, it holds that
    \[ \onenorm{\rho-\sigma} \leq \sqrt{\relent{\rho}{\sigma}} \quad \text{ and } \quad  1-\F(\rho,\sigma)\leq\relent{\rho}{\sigma} . \]
\end{fact}

\begin{fact}
	\label{fact:subsystem monotone}
	The relative entropy is non-increasing when subsystems are considered.
	Let $\rho^{XY}$ and $\sigma^{XY}$ be quantum states, then $\relent{\rho^{XY}}{\sigma^{XY}} \geq \relent{\rho^{X}}{\sigma^{X}}.$
\end{fact}

The following fact is easily verified.
\begin{fact}\label{fact:expecclose}
	Let $0 < \ve, \ve' < 1$, $ 0 < c $,
	$\mu$ and $\mu'$ be probability distributions on a set $\finset{X}$,
	and $f:\finset{X}\rightarrow[0,c]$ be a function.
	If $\expec{x\leftarrow\mu}{f(x)}\leq\ve$ and $\onenorm{\mu-\mu'}\leq\ve'$
	then $\expec{x\leftarrow\mu'}{f(x)}\leq\ve+c\ve'$.
\end{fact}

\subsection*{Useful lemmas}
Here we state and prove some lemmas that we will use later.

\begin{lemma}\label{lem:rmin is small}
Let $\ket{\psi}^{AB}$ be a bipartite pure state with the marginal state on register $B$ being $\rho$. Let a $0/1$ outcome measurement be performed on register $A$ with outcome $O$.  Let $\prob{O=1} = q$. Let the marginal states on register $B$ conditioned on $O=0$ and $O=1$ be $\rho_0$ and $\rho_1$ respectively. Then, $\rminent{\rho_1}{\rho} \leq \log \frac{1}{q}$.
\end{lemma}

\begin{proof}
It is easily seen that $\rho = q \rho_1 + (1-q) \rho_0$. Hence $\rminent{\rho_1}{\rho} \leq \log \frac{1}{q}$.
\end{proof}

\BigCalculation
The following lemma states that when the concerned mutual information is small, then a measurement on Alice's side can be simulated by a unitary operation on Alice's side.
\begin{lemma}
	\label{lem:one sided transformation} Let $\mu$ be a probability distribution on \finset{X}.
	Let
	\[ \ket{\varphi} \defeq \sum_{x \in \finset{X}}
	\sqrt{\fn{\mu}{x}} \ket{xx}^{\xtildesp X} \otimes \ket{\psi_x}^{AB} \]
	be a joint pure state of Alice and Bob, where registers $\xtildesp X A$ are with Alice and register $B$ is with Bob.
	Let $ \mutinf{\register{X}}{B}_{\varphi} \leq \ve $
	and $ \ket{\varphi_x} \defeq \ket{xx} \otimes \ket{\psi_x} $.
	There exist unitary operators $\set{\unitary{U}_x}_{x\in\finset{X}}$
	acting on $\xtildesp X A$ such that
	\[\expec{x\leftarrow\mu}{
		\onenorm{\ketbra{\varphi_x} -
		\br{\unitary{U}_x \otimes \id_{\cspace{B}}}
		\ketbra{\varphi}
		\br{\adjoint{\unitary{U}_x} \otimes \id_{\cspace{B}}}}}\leq4\sqrt{\ve}.\]
\end{lemma}
\begin{proof}
	Let us denote the reduced state of Bob in
	\ket{\varphi_x} and \ket{\varphi} by
	\begin{align*}
		\rho_x \defeq \partrace{A}{\ketbra{\psi_x}}
		\quad \text{and} \quad
		\rho \defeq \partrace{\xtildesp XA}{\ketbra{\varphi}} .
	\end{align*}
	Using \cref{fact:one norm and relent}, it holds that
	\begin{align*}
		\varepsilon & \geq \mutinf{\register{X}}{\register{B}} = \expec{x\leftarrow\mu}{\relent{\rho_x}{\rho}} \geq1-\expec{x\leftarrow\mu}{\F(\rho_x,\rho)}.
	\end{align*}
	By the unitary equivalence of purifications and \cref{thm:Uhlmann},
	there exists a $\unitary{\unitary{U}}_x$
	for each $x\in\finset{X}$ such that
	\[ |\bra{\varphi_x}\br{\unitary{U}_x \otimes \id_{\cspace{B}}}\ket{\varphi}|=\F(\rho_x,\rho). \]
	The lemma follows from the following calculation.
	\begin{align}
		& \expec{x\leftarrow\mu}{
		\onenorm{\ketbra{\varphi_x} -
		\br{\unitary{U}_x \otimes \id_{\cspace{B}}}
		\ketbra{\varphi}
		\br{\adjoint{\unitary{U}_x} \otimes \id_{\cspace{B}}}}} \nonumber \\
		&\qquad = 2 \expec{x\leftarrow\mu}
		{\sqrt{1-|\bra{\varphi_x}\br{\unitary{U}_x \otimes \id_{\cspace{B}}}\ket{\varphi}|^2}}
		\label{eqn:trace distance pure} \\
		&\qquad \leq 2 \sqrt{1-\expec{x\leftarrow\mu}
		{|\bra{\varphi_x}\br{\unitary{U}_x \otimes \id_{\cspace{B}}}\ket{\varphi}|}^2}
		\label{eqn:concavity of sqrt} \\
		&\qquad = 2 \sqrt{1-\expec{x\leftarrow\mu}{\F(\rho_x,\rho)}^2} \nonumber \\
		&\qquad \leq 4 \sqrt{\ve} . \nonumber
	\end{align}
	where \cref{eqn:trace distance pure} follows from \cref{prop:tracefidelityequi}
	and at \cref{eqn:concavity of sqrt} we used the
	concavity of the function $\sqrt{1-\alpha^2}$.
\end{proof}

The following is a generalization of the above lemma that states that when
the concerned mutual informations are small then the simultaneous measurements on
Alice's and Bob's side can be simulated by unitary operations on Alice's and Bob's side.
It is a special case of a more general result in Ref.~\cite{Jain:2008}.

\begin{lemma}[\cite{Jain:2008}]
	\label{lem:two sided transformation}
	Let $\mu$ be a probability distribution over $\finset{X}\times\finset{Y}$.
	Let $\mu_X$ and $\mu_Y$ be the marginals of $\mu$ on $\finset{X}$ and $\finset{Y}$.
	Let
	\[ \ket{\varphi} \defeq \sum_{x \in \finset{X},
		y \in \finset{Y}} \sqrt{\mu(x,y)}
		\ket{xxyy}^{\xtildesp X\ytildesp Y} \otimes \ket{\psi_{x,y}}^{AB} \]
	be a joint pure state of Alice and Bob, where registers $\xtildesp XA$
	belong to Alice and registers $\ytildesp YB$ belong to Bob.
	Let \[ \mutinf{\register{X}}{BY\ytildesp}_{\varphi} \leq \ve
	\quad \text{and} \quad
	\mutinf{\register{Y}}{AX\xtildesp}_{\varphi} \leq \varepsilon . \]
	Let $ \ket{\varphi_{x,y}} \defeq \ket{xxyy} \otimes \ket{\psi_{x,y}} $.
	There exist unitary operators
	$ \cbr{\unitary{U}_x}_{x \in \finset{X}} $
	on \cspace{\xtildesp XA} and $ \cbr{\unitary{V}_y}_{y \in \finset{Y}} $
	on \cspace{\ytildesp YB} such that
	\begin{align*}
		&\expec{(x,y)\leftarrow \mu}{\onenorm{\ketbra{\varphi_{x,y}}
		- \br{\unitary{U}_x \otimes \unitary{V}_y}
		\ketbra{\varphi}
		\br{\adjoint{\unitary{U}}_x \otimes \adjoint{\unitary{V}}_y}}} \\
		& \qquad \leq 8 \sqrt{\ve} + 2\onenorm{\mu - \mu_X \otimes \mu_Y}.
	\end{align*}
\end{lemma}
\begin{proof}
Let $\ket{\varphi_x}$ be the state obtained when we measure register $X$ in $\ket{\varphi}$ and obtain $x$.
Similarly let  $\ket{\varphi_y}$ be the state obtained when we measure register $Y$ in $\ket{\varphi}$ and obtain $y$.
By \cref{lem:one sided transformation}, there exist unitary operators
$\set{\unitary{U}_x}_{x\in\finset{X}}$ and $\set{\unitary{V}_y}_{y\in\finset{Y}}$ such that
\[\expec{x\leftarrow\mu_X}{
		\onenorm{\ketbra{\varphi_x} -
		\br{\unitary{U}_x \otimes \id_{\cspace{B}}}
		\ketbra{\varphi}
		\br{\adjoint{\unitary{U}_x} \otimes \id_{\cspace{B}}}}}\leq4\sqrt{\ve} \]
and
\[\expec{y\leftarrow\mu_Y}{
		\onenorm{\ketbra{\varphi_y} -
		\br{\id_{\cspace{A}} \otimes \unitary{V}_y }
		\ketbra{\varphi}
		\br{\id_{\cspace{A}} \otimes \adjoint{\unitary{V}_y}}}}\leq4\sqrt{\ve}.\]
	Using the above, we get the bound of \cref{eqn:intermediate bound}
	from the calculation that is on the bottom of this page.
	\Cref{eqn:triangle inequality} follows from the triangle
	inequality, the second term in \cref{eqn:Ux} is because $\unitary{U}_x$ doesn't
	change the $\ell_1$-distance, and the first term in \cref{eqn:Ux}
	follows from \cref{prop:monotonequantumoperation} with the
	super-operator that corresponds to measuring \register{Y} in the
	standard basis and storing the outcome in a new register.
	The lemma follows from the following calculation.
	\begin{align*}
		&\expec{(x,y)\leftarrow \mu}{\onenorm{\ketbra{\varphi_{x,y}}
		- \br{\unitary{U}_x \otimes \unitary{V}_y}
		\ketbra{\varphi}
		\br{\adjoint{\unitary{U}}_x \otimes \adjoint{\unitary{V}}_y}}} \\
		& \quad {} = \onenormbigg{\expecbig{(x,y)\leftarrow \mu}
		{\ketbra{xy} \otimes \ketbra{\varphi_{x,y}} \\
		& \quad \quad {} - \ketbra{xy} \otimes \br{\unitary{U}_x \otimes \unitary{V}_y}
		\ketbra{\varphi}
		\br{\adjoint{\unitary{U}}_x \otimes \adjoint{\unitary{V}}_y}}} \\
		& \quad {} \leq \onenormbigg{\expec{(x,y)\leftarrow \mu}{\ketbra{xy}\otimes\ketbra{\varphi_{x,y}}} \\
		& \quad \quad {} - \expecbig{(x,y)\leftarrow \mu_X\otimes\mu_Y}{\ketbra{xy} \\
		& \quad \quad \qquad {} \otimes \br{\unitary{U}_x \otimes \unitary{V}_y}
		\ketbra{\varphi}
		\br{\adjoint{\unitary{U}_x} \otimes  \adjoint{\unitary{V}_y}}}} \\
		& \quad {} + \onenormbigg{\expecbig{(x,y)\leftarrow \mu_X\otimes\mu_Y}{\ketbra{xy} \\
		& \quad \quad \qquad {} \otimes \br{\unitary{U}_x \otimes \unitary{V}_y}
		\ketbra{\varphi} \br{\adjoint{\unitary{U}_x} \otimes \adjoint{\unitary{V}_y}}} \\
		& \quad \quad {} - \expecbig{(x,y)\leftarrow \mu}{\ketbra{xy} \\
		& \quad \quad \qquad {} \otimes \br{\unitary{U}_x \otimes \unitary{V}_y}
		\ketbra{\varphi} \br{\adjoint{\unitary{U}_x} \otimes \adjoint{\unitary{V}_y}}}}
		\label{eqn:triangle inequality 2} \\
		& \quad {} \leq 8 \sqrt{\ve} + 2 \onenorm{\mu - \mu_X \otimes \mu_Y}
	\end{align*}
	where the first inequality follows from the triangle inequality
	and at the last inequality we used \cref{eqn:intermediate bound}
	and \cref{fact:expecclose}.
\end{proof}

\section{Proof of the main result}
\label{sec:main result}

Let a game $G=(\mu,\finset{X},\finset{Y},\finset{A},\finset{B},V)$ be given.
We assume that the distribution $\mu = \mu_X \otimes \mu_Y$
is product across $\finset{X}$ and $\finset{Y}$.
Before the game starts, Alice and Bob share a pure state on the registers
$AE_A'BE_B'$, where $A$ and $B$ are used to store the answers for Alice and Bob, respectively.
After getting the inputs, Alice and Bob perform unitary operations independently
and then they measure registers $A$ and $B$.
The outcomes of the measurements are sent to the referee.
Now, let's consider the game $G^k$.
Let $x = x_1 \ldots x_k \in \finset{X}^k$, $y = y_1 \ldots y_k \in \finset{Y}^k$,
$a= a_1 \ldots a_k \in \finset{A}^k$, and $b = b_1 \ldots b_k \in \finset{B}^k$.
To make notations short, we denote $\mu(x,y)=\prod_i\mu(x_i,y_i)$ and
$V(x,y,a,b) = \prod_i V(x_i,y_i,a_i,b_i)$, whenever it is clear from the context.
Let $\finset{C}\subseteq[k]$ and let $\comp{\finset{C}}$ represent its {\em complement} in $[k]$.
Let $x_{\finset{C}}$ represent the substring of $x$ corresponding to the indices
in $\finset{C}$.
(Similarly, we will use $y_{\finset{C}}, a_{\finset{C}}, b_{\finset{C}}$.)
Let's define
\begin{align*}
	\ket{\startingst} &\defeq \sum_{x,y}\sqrt{\mu(x,y)}
	\ket{xxyy}^{\xtildesp X\ytildesp Y} \\
	& \qquad {} \otimes \sum_{a_{\finset{C}}b_{\finset{C}}}
	\ket{a_{\finset{C}}b_{\finset{C}}}^{A_{\finset{C}}B_{\finset{C}}}
	\otimes\ket{\gamma_{xya_{\finset{C}}b_{\finset{C}}}}^{E_AE_B}
\end{align*}
where $E_A\defeq E_A'A_{\comp{\finset{C}}}$, $E_B\defeq E_B'B_{\comp{\finset{C}}}$, and
$\sum_{a_{\finset{C}}b_{\finset{C}}}\ket{a_{\finset{C}}b_{\finset{C}}}\otimes\ket{\gamma_{xya_{\finset{C}}b_{\finset{C}}}}$
is the shared state after Alice and Bob performed their unitary operations
corresponding to questions $x$ and $y$.
(Note that $\ket{\gamma_{xya_{\finset{C}}b_{\finset{C}}}}$ is unnormalized.)
Consider the state
\begin{align*}
	\ket{\finalst} &\defeq \frac{1}{\sqrt{q}} \sum_{x,y} \sqrt{\mu(x,y)}
	\ket{xxyy}^{\xtildesp X\ytildesp Y} \\
	& \quad {} \otimes \sum_{a_\finset{C}b_{\finset{C}}:
	V(x_{\finset{C}},y_{\finset{C}},a_{\finset{C}},b_{\finset{C}})=1}
	\ket{a_{\finset{C}}b_{\finset{C}}}^{A_{\finset{C}}B_{\finset{C}}}
	\otimes \ket{\gamma_{xya_{\finset{C}}b_{\finset{C}}}}^{E_AE_B}
\end{align*}
where normalizer $q$ is the probability of success on $\finset{C}$.
\begin{lemma}\label{lem:lowrelent}
\begin{align*}
	&\expec{x_{\finset{C}}y_{\finset{C}}a_{\finset{C}}b_{\finset{C}}
	\leftarrow\finalst^{X_{\finset{C}}Y_{\finset{C}}A_{\finset{C}}B_{\finset{C}}}}
	{\relent{\finalst^{\tilde{X}_{\comp{\finset{C}}}\tilde{Y}_{\comp{\finset{C}}}XYE_AE_B}_{x_{\finset{C}}
	y_{\finset{C}}a_{\finset{C}}b_{\finset{C}}}}
	{{\startingst^{\tilde{X}_{\comp{\finset{C}}}\tilde{Y}_{\comp{\finset{C}}}
	XYE_AE_B}_{x_{\finset{C}}y_{\finset{C}}}}}} \\
	& \qquad \leq - \log q + \abs{\finset{C}} \cdot
	\fn{\log}{\abs{\finset{A}} \cdot \abs{\finset{B}}}.
\end{align*}
\end{lemma}
\begin{proof}
Note that, by Lemma~\ref{lem:rmin is small},
\[\rminent{\finalst^{\tilde{X}_{\comp{\finset{C}}}\tilde{Y}_{\comp{\finset{C}}}XYE_AE_B}}
{\startingst^{\tilde{X}_{\comp{\finset{C}}}\tilde{Y}_{\comp{\finset{C}}}XYE_AE_B}}\leq-\log q.\]
Let $p(a_{\finset{C}}, b_{\finset{C}})$ be the probability of obtaining
$(a_{\finset{C}}, b_{\finset{C}})$ when measuring registers $(A_{\finset{C}}, B_{\finset{C}})$ in $\ket{\finalst}$.
Consider,
\begin{align*}
	& \expec{a_{\finset{C}}b_{\finset{C}}\leftarrow\finalst^{A_{\finset{C}}B_{\finset{C}}}}
		{\rminent{\finalst^{\tilde{X}_{\comp{\finset{C}}}\tilde{Y}_{\comp{\finset{C}}}
		XYE_AE_B}_{a_{\finset{C}}b_{\finset{C}}}}
		{\startingst^{\tilde{X}_{\comp{\finset{C}}}\tilde{Y}_{\comp{\finset{C}}}XYE_AE_B}}} \\
	&\qquad \leq \expec{a_{\finset{C}}b_{\finset{C}}\leftarrow\finalst^{A_{\finset{C}}B_{\finset{C}}}}
		{\rminent{\finalst^{\tilde{X}_{\comp{\finset{C}}}\tilde{Y}_{\comp{\finset{C}}}
		XYE_AE_B}_{a_{\finset{C}}b_{\finset{C}}}}
		{\finalst^{\tilde{X}_{\comp{\finset{C}}}\tilde{Y}_{\comp{\finset{C}}}XYE_AE_B}} \right. \\
	& \qquad \qquad {} + \left. \rminent{\finalst^{\tilde{X}_{\comp{\finset{C}}}\tilde{Y}_{\comp{\finset{C}}}XYE_AE_B}}
		{\startingst^{\tilde{X}_{\comp{\finset{C}}}\tilde{Y}_{\comp{\finset{C}}}XYE_AE_B}}}\\
	&\qquad \leq \expec{a_{\finset{C}}b_{\finset{C}}\leftarrow\finalst^{A_{\finset{C}}B_{\finset{C}}}}
		{- \log p(a_{\finset{C}}, b_{\finset{C}}) - \log q} \\
	&\qquad = -\log q + \fn{\mathrm{S}}{\finalst^{A_{\finset{C}}B_{\finset{C}}}} \\
	&\qquad \leq -\log q+|\finset{C}|\cdot\br{\log|\finset{A}|+\log|\finset{B}|}.
\end{align*}
Now,
\begin{align*}
	&- \log q + \abs{\finset{C}} \cdot \fn{\log}{\abs{\finset{A}} \cdot \abs{\finset{B}}} \\
	& \quad \geq \expec{a_{\finset{C}}b_{\finset{C}}\leftarrow\finalst^{A_{\finset{C}}B_{\finset{C}}}}
		{\rminent{\finalst^{\tilde{X}_{\comp{\finset{C}}}\tilde{Y}_{\comp{\finset{C}}}
		XYE_AE_B}_{a_{\finset{C}}b_{\finset{C}}}}
		{\startingst^{\tilde{X}_{\comp{\finset{C}}}\tilde{Y}_{\comp{\finset{C}}}XYE_AE_B}}} \\
	& \quad \geq \expec{a_{\finset{C}}b_{\finset{C}}\leftarrow\finalst^{A_{\finset{C}}B_{\finset{C}}}}
		{\relent{\finalst^{\tilde{X}_{\comp{\finset{C}}}\tilde{Y}_{\comp{\finset{C}}}
		XYE_AE_B}_{a_{\finset{C}}b_{\finset{C}}}}
		{\startingst^{\tilde{X}_{\comp{\finset{C}}}\tilde{Y}_{\comp{\finset{C}}}XYE_AE_B}}} \\
	& \quad \geq \expec{x_{\finset{C}}y_{\finset{C}}a_{\finset{C}}b_{\finset{C}} \\ \leftarrow
		\finalst^{X_{\finset{C}}Y_{\finset{C}}A_{\finset{C}}B_{\finset{C}}}}
		{\relent{\finalst^{\tilde{X}_{\comp{\finset{C}}}\tilde{Y}_{\comp{\finset{C}}}XYE_AE_B}_{x_{\finset{C}}
		y_{\finset{C}}a_{\finset{C}}b_{\finset{C}}}}
		{{\startingst^{\tilde{X}_{\comp{\finset{C}}}\tilde{Y}_{\comp{\finset{C}}}
		XYE_AE_B}_{x_{\finset{C}}y_{\finset{C}}}}}}
\end{align*}
	where the last inequality follows from \cref{fact:relative entropy splitting}.
\end{proof}

	For each $ i \in \Br{k} $, let us define a binary random variable
	$ T_i \in \cbr{0,1} $, which indicates success in the $i$-th repetition.
	That is, $ T_i = \fn{V}{X_i, Y_i, A_i, B_i} $.
	Our main theorem will follow from the following lemma.

	\begin{lemma}\label{lem:mainlemmaproduct}
        Let $0.1 > \delta_1, \delta_2, \delta_3 > 0$ such that
        $\delta_3 = \delta_2 + \delta_1 \cdot \fn{\log}{\abs{\finset{A}} \cdot \abs{\finset{B}}}$.
        Let $ k' \defeq \floor{\delta_1 k} $.
        For any quantum strategy for the $k$-fold game $G^k$, there exists a set
        $\set{i_1, \ldots, i_{k'}}$, such that
	for each $ 1 \leq r \leq k'-1 $, either
		\begin{align*}
			\prob{\Tsr = 1} &\leq 2^{-\delta_2k}
			\intertext{or}
			\prob{ T_{i_{r+1}}=1 \middle\vert \Tsr = 1 }
			&\leq \omega^*(G)+12\sqrt{10 \delta_3}
		\end{align*}
		where $ \displaystyle \Tsr \defeq \prod_{j=1}^{r} T_{i_j} $.
	\end{lemma}
	\begin{proof}
		In the following, we assume that $1 \leq r < k'$.
		However, the same argument also works when $r=0$, \ie{}
		for identifying the first coordinate,
		which we skip for the sake of avoiding repetition.
		Suppose that we have already identified $r$ coordinates
		$ i_1, \ldots, i_r $ satisfying that
		\begin{align*}
			\prob{T_{i_1} = 1} &\leq \omega^*(G)+12\sqrt{10 \delta_3}
			\intertext{and}
			\prob{T_{i_{j+1}} = 1 \middle\vert T^{(j)} = 1}
			&\leq \omega^*(G)+12\sqrt{10 \delta_3}
		\end{align*}
		for $ 1 \leq j \leq r - 1 $.
		If $\prob{\Tsr = 1} \leq 2^{-\delta_2k}$ then we are done,
		so from now on, we assume that $\prob{\Tsr = 1}  >  2^{-\delta_2 k}$.
		Let $\finset{C} \defeq \set{i_1,\ldots, i_r}$.
		To simplify notations, let $\Alice\defeq \xtildesp_{\comp{\finset{C}}} XE_A$, $\Bob\defeq\ytildesp_{\comp{\finset{C}}} YE_B$, and $R_i\defeq X_{\finset{C}}Y_{\finset{C}}X_{<i}Y_{<i}A_{\finset{C}}B_{\finset{C}}$.
		For coordinate $i$, let $\ket{\finalst_{x_{<i}y_{<i}}}$ be the pure state
		that results when we measure registers $X_{<i}Y_{<i}$ (\ie{} registers
		$X_1, \ldots, X_{i-1}, Y_1, \ldots, Y_{i-1}$) in $\ket{\finalst}$
		and get outcome $x_{<i}y_{<i}$.
		We argue now that for a typical coordinate outside $\finset{C}$,
		the distribution of questions is close to $\mu$ in the state $\finalst$.
		We also prove that, for this coordinate, the questions and $R_i$
		are almost independent.
		From \cref{lem:lowrelent}, we get that
		\begin{align}
			\delta_3k &\geq \delta_2 k + \abs{\finset{C}} \cdot
				\fn{\log}{\abs{\finset{A}} \cdot \abs{\finset{B}}} \nonumber \\
			&\geq \expec{x_{\finset{C}}y_{\finset{C}}a_{\finset{C}}b_{\finset{C}} \\
				\leftarrow\finalst^{X_{\finset{C}}Y_{\finset{C}}A_{\finset{C}}B_{\finset{C}}}}
				{\relent{\finalst^{\tilde{X}_{\comp{\finset{C}}}\tilde{Y}_{\comp{\finset{C}}}XYE_AE_B}_{x_{\finset{C}}
				y_{\finset{C}}a_{\finset{C}}b_{\finset{C}}}}
				{{\startingst^{\tilde{X}_{\comp{\finset{C}}}\tilde{Y}_{\comp{\finset{C}}}
				XYE_AE_B}_{x_{\finset{C}}y_{\finset{C}}}}}}
				\label{eqn:relative entropy bound} \\
			&\geq \expec{x_{\finset{C}}y_{\finset{C}}a_{\finset{C}}b_{\finset{C}}\leftarrow
				\finalst^{X_{\finset{C}}Y_{\finset{C}}A_{\finset{C}}B_{\finset{C}}}}
				{\relent{\finalst^{XY}_{x_{\finset{C}}y_{\finset{C}}a_{\finset{C}}b_{\finset{C}}}}
				{{\startingst^{XY}_{x_{\finset{C}}y_{\finset{C}}}}}}
				\label{eqn:relative entropy monotone} \\
			&=\sum_{i\notin\finset{C}}\expec{r_i\leftarrow\finalst^{R_i}}
				{\relent{\finalst^{X_iY_i}_{r_i}}{\startingst^{X_iY_i}}}
				\label{eqn:distcloseproduct}\\
			&=\sum_{i\notin\finset{C}}\relent{\finalst^{X_iY_iR_i}}
				{\finalst^{R_i}\otimes\startingst^{X_iY_i}}
				\label{eqn:split to tensor} \\
			&\geq \sum_{i \notin \finset{C}} \relent{\finalst^{X_iY_iR_i}}
				{\finalst^{R_i} \otimes \finalst^{X_iY_i}}
				\label{eqn:mutual inf relative entropy} \\
			&\geq \sum_{i \notin \finset{C}} \expec{x_i y_i \leftarrow \finalst^{X_iY_i}}
				{\relent{\finalst^{R_i}_{x_i y_i}}
				{\finalst^{R_i}}}
				\label{eqn:only Ri} \\
			&\geq \sum_{i \notin \finset{C}} \expec{x_i y_i \leftarrow \finalst^{X_iY_i}}
				{\onenorm{\finalst^{R_i}_{x_i y_i} - \finalst^{R_i}}^2}
				\label{eqn:go to one norm} \\
			&\geq \sum_{i \notin \finset{C}} \br{\expec{x_i y_i \leftarrow \finalst^{X_iY_i}}
				{\onenorm{\finalst^{R_i}_{x_i y_i} - \finalst^{R_i}}}}^2
 				\label{eqn:publiccoins}
		\end{align}
		where \cref{eqn:relative entropy bound} follows from \cref{lem:lowrelent},
		\cref{eqn:relative entropy monotone} follows from \cref{fact:subsystem monotone},
		\cref{eqn:distcloseproduct,eqn:split to tensor,eqn:only Ri} follow from
		\cref{fact:relative entropy splitting}, \cref{eqn:mutual inf relative entropy}
		follows from \cref{fact:mutinf is min}, \cref{eqn:go to one norm}
		follows from \cref{fact:one norm and relent}, and \cref{eqn:publiccoins}
		follows from the convexity of the function $\alpha^2$.
		Next, we argue that for a typical coordinate outside $\finset{C}$,
		the information between Alice's questions and Bob's
		registers is small in \ket{\finalst}.
		Again, from \cref{lem:lowrelent,fact:subsystem monotone}, we get that
		\begin{align}
			\delta_3 k &\geq \expec{x_{\finset{C}}y_{\finset{C}}a_{\finset{C}}b_{\finset{C}}\leftarrow
				\finalst^{X_{\finset{C}}Y_{\finset{C}}A_{\finset{C}}B_{\finset{C}}}}
				{\relent{\finalst^{X\Bob}_{x_{\finset{C}}y_{\finset{C}}a_{\finset{C}}b_{\finset{C}}}}
				{{\startingst^{X\Bob}_{x_{\finset{C}}y_{\finset{C}}}}}}\nonumber\\
			&\geq \condmutinf{X}{\Bob}{R_1}_{\finalst}
				\label{eqn:go to mutual inf} \\
			&\geq\sum_{i\notin\finset{C}} \condmutinf{X_i}{\Bob}{R_1X_{<i}} _{\finalst}
				\label{eqn:chain rule mutual inf} \\
			&\geq\sum_{i\notin\finset{C}} \condmutinf{X_i}{\Bob}{R_i}_{\finalst}
				\label{eqn:Alice input independent of Bob}
		\end{align}
		where at \cref{eqn:go to mutual inf} we used \cref{fact:mutinf is min}
		and the fact that
		$ \startingst^{X\Bob}_{x_{\finset{C}}y_{\finset{C}}}
		= \startingst^{X}_{x_{\finset{C}}y_{\finset{C}}}
		\otimes \startingst^{\Bob}_{x_{\finset{C}}y_{\finset{C}}}$.
		\Cref{eqn:chain rule mutual inf,eqn:Alice input independent of Bob}
		follow from the chain rule for
		the mutual information and at \cref{eqn:Alice input independent of Bob}
		we also used the observation that \Bob{} contains register $Y$.
		Similarly to the above, for Bob's questions we have
		\begin{align}
			\delta_3 k \geq \sum_{i \notin \finset{C}} \condmutinf{Y_i}{\Alice}{R_i}_{\finalst} .
			\label{eqn:Bob input independent of Alice}
		\end{align}
		From \cref{eqn:distcloseproduct,eqn:publiccoins,eqn:Alice input
		independent of Bob,eqn:Bob input independent of Alice} and using standard application of Markov's inequality,
		we get that there exists a coordinate $ j \notin \finset{C} $ such that
		\begin{align}
			\expec{r_j\leftarrow\finalst^{R_j}}
				{\relent{\finalst^{X_jY_j}_{r_j}}{\startingst^{X_jY_j}}}
				\leq \frac{5 \delta_3}{1- \delta_1}
				&\leq 10\delta_3
				\label{eqn:eq2}\\
			\expec{x_j y_j \leftarrow \finalst^{X_j Y_j}}
				{\onenorm{\finalst^{R_j}_{x_j y_j} - \finalst^{R_j}}}
				\leq \sqrt{\frac{5 \delta_3}{1- \delta_1}}
				&\leq \sqrt{10\delta_3}
				\label{eqn:eq1}\\
			\condmutinf{X_j}{\Bob}{R_j}_{\finalst}
				\leq \frac{5 \delta_3}{1- \delta_1}
				&\leq 10\delta_3
				\label{eqn:eq3}\\
			\condmutinf{Y_j}{\Alice}{R_j}_{\finalst}
				\leq \frac{5 \delta_3}{1- \delta_1}
				&\leq 10\delta_3 .
				\label{eqn:eq4}
		\end{align}
		Let $\ket{\finalst_{r_j}}$ be the pure state that we get when we measure
		register $R_j$ in $\ket{\finalst}$ and get outcome $r_j$.

		Suppose that there exists a protocol $\Prot_0$ for $G^k$
			which wins all coordinates in $\finset{C}$ with probability
			greater than $2^{-\delta_2k}$.
			Moreover, conditioning on success on all coordinates in
			$\finset{C}$, the probability it wins the game in
			the $j$-th coordinate is $\omega$.

		\begin{itemize}
			\item Let us construct a new protocol $\Prot_1$, that starts with
			the joint state $\finalst^{X_jY_jR_jE_AE_B}$, where $X_jE_A$ and $Y_jE_B$
			are given to Alice and Bob, respectively, and $R_j$ is shared between them.
			From our assumption, the probability that Alice and Bob win the game
			in the $j$-th coordinate is $\omega$.
			\item Let us consider a new protocol $\Prot_2$, where Alice and Bob
			are given questions $(x_j,y_j) \leftarrow \finalst^{X_jY_j}$ and
			they share $r_j \leftarrow \finalst^{R_j}_{x_jy_j}$ as public coins.
			By \cref{lem:two sided transformation}, they are able to create
			a joint state that is close to the starting state of $\Prot_1$
			by sharing $\ket{\finalst_{r_j}}$ and applying local unitary operations.
			More concretely, \cref{eqn:eq3,eqn:eq4} show the conditions
			for the mutual informations required by \cref{lem:two sided transformation}.
			From \cref{eqn:eq2}, we can get
			\begin{align}
				10 \delta_3 &\geq
					\expec{r_j\leftarrow\finalst^{R_j}}
					{\relent{\finalst^{X_jY_j}_{r_j}}{\startingst^{X_jY_j}}}
					\nonumber \\
				&\geq \expec{r_j\leftarrow\finalst^{R_j}}
					{\relent{\finalst^{X_jY_j}_{r_j}}
					{\finalst^{X_j}_{r_j} \otimes \finalst^{Y_j}_{r_j}}}
					\label{eqn:split to tensor 2} \\
				&\geq \expec{r_j\leftarrow\finalst^{R_j}}
					{\onenorm{\finalst^{X_jY_j}_{r_j} -
					\finalst^{X_j}_{r_j} \otimes \finalst^{Y_j}_{r_j}}^2}
					\label{eqn:go to one norm 2} \\
				&\geq \br{\expec{r_j\leftarrow\finalst^{R_j}}
					{\onenorm{\finalst^{X_jY_j}_{r_j} -
					\finalst^{X_j}_{r_j} \otimes \finalst^{Y_j}_{r_j}}}}^2
					\label{eqn:convexity}
			\end{align}
			where \cref{eqn:split to tensor 2} follows from \cref{fact:mutinf is min},
			\cref{eqn:go to one norm 2} follows from \cref{fact:one norm and relent},
			and at \cref{eqn:convexity} we used the convexity of the function $\alpha^2$.
			This implies
			\[ \expec{r_j\leftarrow\finalst^{R_j}}
			{\onenorm{\finalst^{X_jY_j}_{r_j}-\finalst^{X_j}_{r_j}\otimes\finalst^{Y_j}_{r_j}}}
			\leq\sqrt{10\delta_3}.\]
			Thus, using the above and~\cref{lem:two sided transformation},
			we conclude that they can win the game with probability at least $\omega-10\sqrt{10\delta_3}$.
			\item Let us construct a new protocol $\Prot_3$, where Alice and Bob are given
			questions $(x_j,y_j) \leftarrow \finalst^{X_jY_j}$.
			They share public coins $r_j \leftarrow \finalst^{R_j}$ and execute the same
			strategy as in $\Prot_2$.
			By \cref{eqn:eq1}, the probability that they win the game is at least
			$\omega-11\sqrt{10\delta_3}$.
			\item Let us consider a new protocol $\Prot_4$, where Alice and Bob
			are given questions $(x,y) \leftarrow \mu$
			and they execute the same strategy as in $\Prot_3$.
			By \cref{eqn:eq2,fact:one norm and relent}, the probability that they
			win the game is at least $\omega-12\sqrt{10\delta_3}$.
			Note that $\Prot_4$ is a strategy for game $G$ under distribution $\mu$.
			This means that $\omega-12\sqrt{10\delta_3}\leq\omega^*(G)$.
		\end{itemize}
		We conclude the lemma.
	\end{proof}

	We can now prove our main result.
	We restate it here for convenience.
\begin{maintheorem}
	Let $\ve > 0$.
	Given a game $G$ with value $\omega^*(G) \leq 1 - \ve$, it holds that
	\begin{align*}
		\fn{\omega^*}{G^k} &\leq \br{1-\frac{\ve}{2}}^{\frac{\ve^2k}
		{12000\br{\log|\finset{A}|+\log|\finset{B}|}}} \\
		&= \br{1-\ve^3}^{\bigomega{\frac{k}{\log|\finset{A}|+\log|\finset{B}|}}}.
	\end{align*}
\end{maintheorem}
\begin{proof}
	We set $\delta_1=\frac{\ve^2}{12000(\log|\finset{A}|+\log|\finset{B}|)}$,
	$\delta_2=\frac{\ve^2}{12000}$, and $\delta_3=\frac{\ve^2}{6000}$.
	Given any strategy for $G^k$, using \cref{lem:mainlemmaproduct},
	either $\fn{\omega^*}{G^k} \leq 2^{-\delta_2k}$, or there are
	\floor{\delta_1 k} coordinates $\set{i_1,\ldots,i_{\floor{\delta_1 k}}}$
	such that the probability Alice and Bob win the $i_j$-th coordinate,
	conditioning on success on all the previous coordinates,
	is at most $1-\ve/2$.
	This finishes the proof of the theorem.
\end{proof}

\section*{Acknowledgements}

	We thank Oded Regev, Thomas Vidick, and anonymous referees for useful comments.
	This work is supported by the Singapore Ministry of Education
	Tier 3 Grant and the Core Grants of the Center for Quantum
	Technologies, Singapore.

\bibliographystyle{IEEEtran}
\bibliography{references}

\end{document}